\documentclass{llncs}

\usepackage{latexsym}

\newtheorem{assumption}{Assumption}

\begin{document}
\title{Finite countermodels for safety verification of parameterized tree systems}
\author{Alexei Lisitsa}

\institute{Department of Computer Science
\\University of Liverpool\\
\texttt{a.lisitsa@liverpool.ac.uk}}

\maketitle
 
\begin{abstract}

In this paper we deal with verification of safety properties of parameterized systems with a tree topology. 
The verification problem is translated to  a purely logical problem of finding a finite 
countermodel for a first-order formula, which further resolved by a generic finite model finding procedure.
A finite countermodel method is shown is at least as powerful as regular tree model checking and as the methods based 
on monotonic abstraction and backwards symbolic reachability. The practical efficiency of the method is illustrated on a set of examples taken from the literature.

 \end{abstract}



\section{Finite Countermodel Method} 
The development of general automated methods for the 
verification of infinite-state and parameterized systems poses a major challenge. In general, such problems 
are undecidable, so one cannot hope for the ultimate solution and the  development should focus on the restricted 
classes of systems and properties.   

In this paper we deal with a very general method for verification of \emph{safety} properties of infinite-state systems  which is based  on a  simple idea.  If an evolution of a computational system is faithfully 
modeled by a derivation in a classical first-order logic then safety verification (non-reachability of unsafe states)  can be reduced to the disproving of a first-order formula. The latter task can be  (partially, at least) tackled by generic automated procedures searching for \emph{finite} countermodels.  

Such an approach to verification was originated in the research on formal verification of security protocols 
(\cite{Weid99,S01,GL08,JW09,Gut}) and later has been extended to the wider classes of infinite-state and parameterized verification tasks. Completeness of the approach for  particular classes of systems (lossy channel systems) and relative completeness with respect to general method of regular model checking has been established in \cite{AL10} and \cite{AL10arxiv} respectively. The method has also been applied to the verification of safety properties of general term rewriting systems and its relative completeness with respect to the tree completion techniques has been shown in \cite{AL11}. 

In this paper we continue investigation of applicability of the method and show its power in the context of verification of safety properties of \emph{parameterized tree-like systems}. We show the relative completeness of FMC methods with 
respect to  \emph{regular tree model checking} \cite{RTMC} and with respect to the methods based on monotonic abstraction and symbolic backwards reachability analysis \cite{PTS}.

\subsection{Preliminaries}~\label{sec:pre}
We assume that the reader is familiar with the basics of first-order logic. In particular, we use without definitions the 
following concepts: first-order predicate logic, first-order models, interpretations of relational, functional and constant symbols, 
satisfaction $M \models \varphi$ of a formula $\varphi$  in a model $M$, semantical consequence $ \varphi\models\psi$, deducibility (derivability) $\vdash$ in first-order logic. We denote interpretations  by square brackets, so, 
for example, $[f]$ denotes an interpretation of a functional symbol $f$ in a model.
We also use the existence of \emph{complete} finite model finding procedures for the first-order predicate 
logic \cite{Model,McCune}, which given a first-order sentence $\varphi$  eventually produce a finite model for $\varphi$ if such a model exists.

\section{Regular Tree Model Checking}

Regular Tree Model Checking (RTMC) is a general method  for the verification of parameterized systems that have tree topology \cite{RTMC,ARTMC}. The definitions of this section are largely borrowed from \cite{RTMC}. 

\subsection{Trees} 
A \emph{ranked alphabet} is a pair $(\Sigma,\rho)$, where $\Sigma$ is a finite set of symbols and $\rho : \Sigma \rightarrow Nat$ is an \emph{arity mapping}. Let $\Sigma_{p}$ denote the set of symbols in $\Sigma$ of arity $p$.  
Intuitively, each node of a tree is a labeled with a symbol from $\Sigma$ and the out-degree of the node is the same as the arity of the symbol. 

\begin{definition}~\label{def:tree} 
A tree $T$ over a ranked alphabet $(\Sigma, \rho)$ is a pair $(S,\lambda)$, where 

\begin{itemize} 
\item $S$, called tree structure, is a finite set of finite sequences over $Nat$. Each sequence $n$ in $S$ is called a \emph{node} of $T$. $S$ is prefix-closed set, that is, if $S$ contains a node $n = b_{1}b_{2}\ldots b_{k}$, then $S$ also contains the node $n' = b_{1}b_{2}\ldots b_{k-1}$ and the nodes $n_{r} = b_{1}b_{2}\ldots b_{k-1}r$, for $r: 0 \le r < b_{k}$. We say that $n'$ is a parent of $n$, and that $n$ is a child of $n'$. A \emph{leaf} of $T$ is a node $n$ which does not have any child.   
\item $\lambda$ is a a mapping from $S$ to $\Sigma$. the number of children of $n$ is equal to $\rho(\lambda(n))$. In particular, if $n$ is a leaf then $\lambda(n) \in \Sigma_{0}$.  
\end{itemize} 
\end{definition}  

We use $T(\Sigma)$ to denote the set of all trees over $\Sigma$. We write $n \in T$ when $n \in S$ and $f \in T$ denotes that $\lambda(n) = f$ for some 
$n \in T$. For a tree $T = (S, \lambda)$ and a node $n \in T$, a subtree of $T$ rooted at $n$ is a tree 
$T' = (S', \lambda_{n})$, where 
$S' \subseteq \{b \mid nb \in S\}$ and $\lambda_{n}(b) = \lambda(nb)$. Notice, that according to 
this definition  a subtree of a tree T consists not necessarily all descendants of some node in T. 

For a ranked alphabet $\Sigma$ let $\Sigma^{\bullet}(m)$ be the ranked alphabet which contains all tuples $(f_{1}, \ldots, f_{m})$ such that $m \ge 1$ and 
$f_{1}, \ldots, f_{m}\in \Sigma_{p}$ for some $p$. We put $\rho((f_{1}, \ldots, f_{m})) = \rho(f_{1})$. 

For trees $T_{1} = (S_{1}, \lambda_{1})$ and $T_{2} = (S_{2}, \lambda_{2})$ we say that $T_{1}$ and $T_{2}$ are structurally equivalent, if $S_{1} = S_{2}$. 

Let $T_{1} = (S, \lambda_{1}), \ldots, T_{m} = (S,\lambda_{m})$  are structurally equivalent trees. Then $T_{1} \times \ldots  \times T_{m}$ denotes the tree $T = (S,\lambda)$ where $\lambda(n) = (\lambda_{1}(n), \ldots, \lambda_{m}(n)$).

\subsection{Tree Automata and Transducers} 

A \emph{tree language} is a set of trees. 

\begin{definition}
A tree automaton over a ranked alphabet $\Sigma$ is a triple $A = (Q,F,\delta)$, where $Q$ is a finite set of states, $F\subseteq Q$ is a set of final states, and 
$\delta$ is a transition relation, represented by a finite set of rules of the form $(q_{1}, \ldots, q_{p}) \rightarrow^{f} q$, where $f \in \Sigma_{p}$ and 
$q_{1}, \ldots q_{p},q \in Q$. 
\end{definition}

A \emph{run} $r$ of $A$ on a tree $T = (S,\lambda) \in T(\Sigma)$ is a mapping from $S$ to $Q$ such that for each node $n \in T$ with children $n_{1}, \ldots, n_{k}$: 
$(r(n_{1}, \ldots, r(n_{k})) \rightarrow^{\lambda (n)} r(n)) \in \delta$. 

For a state $q \in S$ we denote by $T \Rightarrow_{A}^{r} q$ that $r$ ia  run of $A$ on $T$ such that $r(\epsilon) = q$. 
We say that $A$ accepts $T$ if $T \Rightarrow_{A}^{r} q$ for some run $r$ and some $q \in F$. The language of trees accepted by an automaton $A$ is defined
 as $L(A) = \{T \mid T \;\mbox{is accepted by}\; A\}$. The tree language $L$ is called \emph{regular} iff there is a tree automaton $A$ such that 
$L = L(A)$.  

A tree automaton over an alphabet $\Sigma^{\bullet}(2)$ is called \emph{tree transducer}.

Let $D$ be a tree transducer over an alphabet $\Sigma^{\bullet}(2)$. 

An one-step transition relation $R_{D}  \subseteq T(\Sigma) \times T(\Sigma)$ is defined as $R^{D} = \{(T,T') \mid T \times T' \; \mbox{ is accepted by}\; D  \}$. The reflexive and transitive closure of 
$R_{D}$ is denoted by $R_{D}^{\ast}$.

We use $\circ$ to denote the composition of two binary relations defined in the standard way. Let $R^{i}$ denote the $i$th power of $R$ i.e. $i$ compositions of $R$. Then we have $R^{\ast} = \cup_{i \ge 0} R^{i}$.   

For any $L \subseteq T(\Sigma)$ and $R \subseteq T(\Sigma) \times T(\Sigma)$ we denote by $L \star R$ the set $\{y \mid  \exists x (x,y) \in L \times T(\Sigma) \cap R  \}$.

Regular Tree Model Checking deals with the following basic verification task. 

\begin{problem}\label{problem:rtmc}
Given two tree automata $A_{I}$ and $A_{U}$ over an alphabet $\Sigma$ and a tree transducer $D$ over $\Sigma^{\bullet}(2)$.  
Does $(L(A_{I}) \star R^{\ast}_{D}) \cap L(A_{U}) = \emptyset$ hold? 
\end{problem}

In verification scenario, trees over $\Sigma$ denote states of the system to be verified, 
tree automata $A_{I}$ and $A_{U}$ define the sets of trees representing \emph{initial}, respectively, \emph{unsafe} states.   Tree transducer $D$ defines the transitions of the system. Under such assumptions, the positive answer to an instance of  Problem~\ref{problem:rtmc} means the \emph{safety} property is established, namely, none of the unsafe states is reachable along the system transitions  from any of the 
initial states.

The verification in RTMC proceeds by producing a tree transducer $TR$ approximating $R^{\ast}_{D}$ from above,    that is 
$R^{\ast}_{D} \subseteq L(TR)$,  and showing the emptiness of the set $(L(A_{I}) \star L(TR)) \cap L(A_{U})$


\section{From RTMC to FMC} 

In this section we show that the generic regular tree model checking question posed in Problem~\ref{problem:rtmc}   
can be reduced to a purely logical problem of finding a finite countermodel for a first-order logic formula, which then can be resolved by application of generic model finding procedure. 
We show also the relative completeness of finite countermodel method with respect to RTMC. 

Assume we are given an instance of the basic verification problem (over ranking alphabet $\Sigma$), that is 

\begin{itemize}
\item a tree automaton $A_{I} = (Q_{I},F_{I}, \delta_{I})$ accepting a regular set of initial states; 
\item a tree automaton $A_{U} = (Q_{U}, F_{U}, \delta_{U} )$ accepting a regular set of unsafe states; 
\item a tree transducer $D = (Q_{D}, F_{D}, \delta_{D} )$ representation one-step transition relation $R_{D}$. 
\end{itemize}

Now define a set formulae of first-order predicate logic as follows. The vocabulary consists of 

\begin{itemize}
\item constants for all elements of  $Q_{I} \sqcup  Q_{U} \sqcup Q_{D} \sqcup  \Sigma_{0}$;
\item unary predicate symbols $Init^{(1)}$, $Unsafe^{(1)}$
\item binary predicate symbols $Init^{2}$, $Unsafe^{2}$, $R$; 
\item a ternary predicate symbol $T$; 
\item a $p$-ary functional symbol $f_{\theta}$ for every $\theta \in \Sigma_{p}$
\end{itemize}

Given any tree  $\tau$ from $T(\Sigma)$ define its term translation $t_{\tau}$ by induction: 

\begin{itemize}
\item $t_{\tau} = c$ for a tree $\tau$ with one node labeled by $c \in \Sigma_{0}$; 
\item $t_{\tau} = f_{\theta}(t_{\tau_{1}}, \ldots, t_{\tau_{p}})$ for a tree $\tau$ with the root labeled by $\theta \in \Sigma_{p}$ and children $\tau_{1}, \ldots \tau_{p}$.  
\end{itemize}

Let $\Phi$ be the set of the following formulae, which are all assumed to be universally closed: 

\begin{enumerate}
\item $Init^{(2) }(a,q)$ for every $a \in \Sigma_{0}$, $q \in Q_{I}$ and $\rightarrow^{a} q$ in $\delta_{I}$; 
\item $Init^{(2) }(x_{1},q_{1}) \land \ldots \land Init^{(2) }(x_{p},q_{p}) \rightarrow Init^{(2) }(f_{\theta}(x_{1},\ldots,x_{p}),q)$ for every 
$(q_{1},\ldots,q_{p}) \rightarrow^{\theta} q$ in $\delta_{I}$; 
\item $\vee_{q \in F_{I}} Init^{(2) }(x,q) \rightarrow Init^{(1)}(x)$;  
\item $Unsafe^{(2) }(a,q)$ for every $a \in \Sigma_{0}$, $q \in Q_{U}$ and $\rightarrow^{a} q$ in $\delta_{U}$; 
\item $Unsafe^{(2) }(x_{1},q_{1}) \land \ldots \land Unsafe^{(2) }(x_{p},q_{p}) \rightarrow Unsafe^{(2) } (f_{\theta}(x_{1},\ldots,x_{p}),q)$ for every 
$(q_{1},\ldots,q_{p}) \rightarrow^{\theta} q$ in $\delta_{U}$; 
\item $\vee_{q \in F_{U}} Unsafe^{(2) }(x,q) \rightarrow Unsafe^{(1)}(x)$;  
\item $T(a,b,q)$ for every $\rightarrow^{(a,b)} q$ in $\delta_{D}$; 
\item $T(x_{1},y_{1},q_{1}) \land \ldots \land T(x_{p},y_{p},q_{p}) \rightarrow T(f_{\theta_{1}}(x_{1}, \ldots, x_{p}), 
f_{\theta_{2}}(y_{1}, \ldots, y_{p}),q)$  for every $(q_{1}, \ldots, q_{p})\rightarrow^{\theta_{1},\theta_{2}} q$ in $\delta_{D}$; 
\item $\vee_{q \in F_{D}} T(x,y,q) \rightarrow R(x,y)$;
\item $R(x,x)$; 
\item $R(x,y) \land R(y,z) \rightarrow R(x,z)$. 
\end{enumerate}

\begin{proposition}~\label{prop:init_unsafe}(adequacy of Init and Unsafe translations)

If $\tau \in L(A_{I})$ then $\Phi \vdash Init^{(1)}(t_{\tau})$

If $\tau \in L(A_{U})$ then $\Phi \vdash Unsafe^{(1)}(t_{\tau})$

\end{proposition}

\begin{proof} We prove only the first statement, the second one is dealt with in the same way. 

\begin{lemma}\label{lemma:}
For any tree $\tau$ and any run $r$ if $\tau \Rightarrow^{r}_{A_{I}} q$ then $\Phi \vdash Init^{2}(t_{\tau},q)$. 
\end{lemma}

{\it Proof of Lemma}. By induction on the depth of the trees.  

\begin{itemize}

\item \emph{Induction Base Case}.  Assume  
$\tau$ has a depth $0$, that is consists of one vertex labeled by some $a \in \Sigma_{0}$. 
Let $r$ be a run such that $\tau \Rightarrow^{r}_{A_{I}} q$. It follows (by the definition of run)  
that $\rightarrow^{a} q \in \delta_{I}$ and then $Init^{(2)}(a,q)$ is in $\Phi$ (by clause 1 of the definition of $\Phi$) and therefore $\Phi \vdash Init^{(2)}(a,q)$ 
Finally notice that term translation $t_{\tau}$ of $\tau$ is $a$.  

\item \emph{Induction Step Case}.  Assume  $\tau$ has a root labeled by $\theta \in \Sigma_{p}$ and  
$\tau_{1}, \ldots, t_{p}$ are children of the root. For  a run $r$ on $\tau$, 
assume $\tau \Rightarrow^{r}_{A_{I}} q$ and  $\tau_{1} \Rightarrow^{r}_{A_{I}} q_{1}, 
\ldots, \tau_{p} \Rightarrow^{r}_{A_{I}} q_{n}$. By the definition of a run we  have   
$(q_{1}, \ldots, q_{p}) \rightarrow^{\theta} q$ is in $\delta_{I}$. By induction assumption we have 
$\Phi \vdash Init^{(2)}(t_{\tau_{1}},q_{1}), \ldots, \Phi \vdash Init^{(2)}(t_{\tau_{p}},q_{p})$. 
By using clause 2 of the definition  of $\Phi$ we get \\
$\Phi \vdash Init^{(2)}(f_{\theta}(t_{\tau_{1}}, \ldots, t_{\tau_{p}}),q)$ and $\Phi \vdash Init^{2}(t_{\tau},q)$
$\Box$

\end{itemize}

Returning to the proof of the proposition we notice that if $\tau \in L(A_{I})$ then there is a run $r$ such 
that $\tau^{r}_{A_{I}} q$ for some $q \in F_{I}$. By Lemma~\ref{lemma:} $\Phi \vdash \vee_{q \in F_{I}} 
Init^{(2)}(t_{\tau}, q)$. By using clause 3 of the definition of $\Phi$ we then get $\Phi \vdash Init^{(1)}(t_{\tau})$
\end{proof}

\begin{proposition}~\label{prop:encoding}(adequacy of encoding) 

If $\tau \in L(A_{I}) \star R^{\ast}_{D}$ then $\Phi \vdash \exists x \; Init^{(1)}(x) \land R(x,t_{\tau})$

\end{proposition} 

\begin{proof} Easy induction on the length of transition sequences. 
\begin{itemize}
\item \emph{Induction Base Case.} Let $\tau \in L(A_{I}) \subseteq L(A_{I}) \star R^{\ast}_{D}$. Then
 $\Phi \vdash Init^{(1)}(t_{\tau})$ (by Proposition~\ref{prop:init_unsafe}) and, further 
$\Phi \vdash \exists x \; Init^{(1)}(x) \land R(x,t_{\tau})$ (using clause 10).  
\item \emph{Induction Step Case.} Let $\tau \in L(A_{I}) \star R^{n+1}_{D}$. Then there exists $\tau'$ such that
$\tau' \in L(A_{I}) \star R^{n}_{D}$ and $R(\tau',\tau)$ holds. Further,  by the argument analogous to the proof of Proposition~\ref{prop:init_unsafe}, $R(\tau',\tau)$ entails $\Phi \vdash \vee_{q \in F_{D}}T(t_{\tau'},t_{\tau},q)$ 
and further  $\Phi \vdash R(t_{\tau'},t_{\tau})$ (using clause  9). from this, the clause 11 and the induction assumption 
$\Phi \vdash \exists x \; Init^{(1)}(x) \land R(x,t_{\tau})$ follows.

\end{itemize}

Assume $\tau \in L(A_{I}) \star R^{\ast}_{D}$ then by definition of $\star$ there exists $\tau_{0} \in L(A_{I})$
such that $R_{D}^{\ast}(\tau_{0},\tau)$ holds.   
\end{proof}

\begin{corollary}~\label{cor:verification}(correctness of the verification method)

If $\Phi \not\vdash \exists x \exists y (Init^{(1)}(x) \land R(x,y) \land Unsafe^{(1)}(y)$ 
then 
$(L(A_{I}) \star R^{\ast}_{D}) \cap L(A_{U})=\emptyset$
\end{corollary}

The corollary~\ref{cor:verification} serves as a formal underpinning of the proposed FCM (finite countermodel) verification method. In order to prove safety, that is $(L(A_{I}) \star R^{\ast}_{D}) \cap L(A_{U})=\emptyset$ it is sufficient to demonstrate $\Phi \not\vdash \exists x \exists y (Init^{(1)}(x) \land R(x,y) \land Unsafe^{(1)}(y)$. In the FCM method we delegate this task to the generic finite model finding procedure, which searches for the finite countermodels for \\
$\Phi \rightarrow  \exists x \exists y (Init^{(1)}(x) \land R(x,y) \land Unsafe^{(1)}(y)$.

\subsection{Relative completeness of FCM with respect to RTMC}

In general, searching for finite countermodels to disprove non-valid first-order formulae may not always lead to success, because for some formulae countermodels are inevitably infinite. Here we show, however,  this is not the case for the first-order encodings of the problems which can be positively answered by Regular Tree Model Checking. It follows then that FCM is at least as powerful in establishing safety as RTMC, provided a complete finite model finding procedure is used.

\begin{theorem}~\label{th:completeness}(relative completeness of FCM)
Given an instance of the basic verification problem for RTMC, that is two tree automata $A_{I}$ and $A_{U}$ 
over an alphabet $\Sigma$ and a tree transducer $D = (Q_{D},F_{D},\delta_{D})$ over $\Sigma^{\bullet}(2)$.  
If  there exists a regular tree language ${\cal R}$ such that $(L(A_{I}) \star R^{\ast}_{D}) \subseteq {\cal R}$ and  
${\cal R} \cap L(A_{U}) = \emptyset$ then there is a finite countermodel for $\Phi \rightarrow \exists x \exists y (Init^{(1)}(x) \land R(x,y) \land Unsafe^{(1)}(y)$  

\end{theorem} 

\begin{proof}   Let $A = (Q,F,\delta)$ be a \emph{deterministic} tree automaton recognizing the tree language ${\cal R}$, i.e. $L(A) = {\cal R}$.   
We take $Q \cup Q_{I} \cup Q_{U} \cup Q_{D} \cup \{ e \}$ to be domain of the required finite model. Here  $e$ is a distinct element not in $Q \cup Q_{I} \cup Q_{U} \cup Q_{D}$. 

Define interpretations as follows. 

\begin{itemize}
\item For $a \in \Sigma_{0}$ $[a] = q \in Q$ such that $\rightarrow^{a} q$ is in $\delta$;  
\item For $\theta \in \Sigma_{p}$ $[f_{\theta}](q_{1}, \ldots, q_{p}) = q$ for any $(q_{1}, \ldots, q_{p}) \rightarrow^{\theta} q$ in $\delta$, and $[f_{\theta}](\ldots) = e$ otherwise;    
\item Interpretations of $Init^{2}$ and $Init^{1}$ are defined inductively, as the least subsets of pairs, respectively, elements of the domain, satisfying the formulae $(1)$ - $(3)$ (and assuming all interpretations above);
\item Interpretations of $Unsafe^{2}$ and $Unsafe^{1}$ are defined inductively, as the least subsets of pairs, respectively, elements of the domain, satisfying the formulae $(4)$ - $(6)$ (and assuming all interpretations above);
\item Interpretation of $T$ is  defined inductively, as the least subsets of triples satisfying the formulae $(7)$ - $(8)$ (and assuming all interpretations above); 
\item Interpretation of $R$ and $Init^{1}$ is  defined inductively, as the least subsets of pairs, satisfying the formulae $(9)$ - $(11)$ (and assuming all interpretations above);
\end{itemize}

Such defined a finite model satisfies $\Phi$ (by construction). 
Now we  check that  $\neg \exists x \exists y (Init^{(1)}(x) \land R(x,y) \land Unsafe^{(1)}(y)$ is satisfied in the model 
We have 
\begin{enumerate}
\item $[Init^{(1)}] \star [R] \subseteq \{[t] \mid t \in L(A_{I}) \star R^{\ast}_{D}\}$ (by the minimality condition on interpretations of $Init^{(1)}$ and $R$); 
\item $\{[t]\mid t \in L(A_{I}) \star R^{ast}_{D}\} \subseteq F \subseteq Q$ (by interpretations of terms and condition 
$(L(A_{I}) \star R^{\ast}_{D}) \subseteq {\cal R}$);
\item $[Init^{(1)}] \star [R] \subseteq F$ (by 1 and 2); 
\item $[Unsafe^{(1)}] = \{[t] \mid t \in L(A_{U})\}$ (by definition of $[Unsafe^{1}]$, in particular by the  minimality condition);
\item  $\{[t] \mid t \in L(A_{U})\} \cap F = \emptyset$ (by condition ${\cal R} \cap L(A_{U}) = \emptyset$); 
\item $Unsafe^{(1)} \cap F = \emptyset$ (by 4 and 5); 
\item   $[Init^{(1)}] \star [R] \cap Unsafe^{(1)} = \emptyset$ (by 3 and 6); 
\end{enumerate} 

\end{proof}

\section{The case study}

In this section we illustrate FCM method by applying it to the verification of Two-way Token protocol. The system consists of finite-state processes connected to form a binary tree structure. Each process stores a single bit which represents the fact that  the process has a token. During operation of the protocol the token can be passed up or down the tree. The correctness condition is that no two or more tokens ever appear. In parameterized verification we would like to establish correctness for all possible sizes of trees. 

We take RTMC-style  specification of Two-way Token  from \cite{RTMC}. Let $\Sigma = \{t,n,T,N\}$ be the alphabet.
Here $t,n \in \Sigma_{0}$ label processes on the leaves of a tree, and $T,N \in \Sigma_{2}$ label processes on the inner nodes of a tree. Further, $t,T$ label processes with a token and $n,N$ label processes without tokens.  

The automaton $A_{I} = (Q_{I},F_{I},\delta_{I})$ accepts the initial configurations of the protocol, that is the trees with exactly one token.  Here 
$Q_{I} = \{q_{0}, q_{1}\}$, $F_{I} = \{q_{1}\}$ and $\delta_{I}$ consists of the following transition rules: 
\begin{center}
  \begin{tabular}{ r  r  r }
     $\rightarrow^{n} q_{0}$ & $\rightarrow^{t} q_{1}$ &  $(q_{0},q_{0}) \rightarrow^{T} q_{1}$  \\
     $(q_{0},q_{0}) \rightarrow^{N} q_{0}$ & $(q_{0},q_{1}) \rightarrow^{N} q_{1}$ & 
     $(q_{1},q_{0}) \rightarrow^{N} q_{1}$ 
    \end{tabular}
\end{center}

The tree transducer $D = (Q_{D},F_{D},\delta_{D})$ over $\Sigma^{\bullet}(2)$ represents the transitions of the protocol. 
Here $Q_{D} = \{q_{0},q_{1},q_{2},q_{3}\}$, $F = \{q_{2}\}$ and $\delta_{D}$ consists of the following transition rules:

\begin{center}
  \begin{tabular}{ r  r }
     $\rightarrow^{(n,n)} q_{0}$ & $\rightarrow^{(t,n)} q_{1}$ \\ 
     $\rightarrow^{(n,t)} q_{3}$ & $(q_{0},q_{0}) \rightarrow^{(N,N)} q_{0}$\\  
     $(q_{0},q_{2}) \rightarrow^{(N,N)} q_{2}$ & $(q_{2},q_{0}) \rightarrow^{(N,N)} q_{2}$ \\ 
     $(q_{0},q_{0}) \rightarrow^{(T,N)} q_{1}$ & $(q_{3},q_{0}) \rightarrow^{(T,N)} q_{2}$ \\
     $(q_{0},q_{3}) \rightarrow^{(T,N)} q_{2}$ & $(q_{0},q_{1}) \rightarrow^{(N,T)} q_{2}$ \\ 
     $(q_{1},q_{0}) \rightarrow^{(N,T)} q_{2}$ & $(q_{0},q_{0}) \rightarrow^{(N,T)} q_{3}$ \\  
    \end{tabular}
\end{center}

The automaton $A_{U} = (Q_{U},F_{U},\delta_{U})$ accepts unsafe (bad)  configurations of the protocol, that is the trees with at least two tokens.  Here 
$Q_{U} = \{q_{0}, q_{1}, q_{2}\}$, $F_{U} = \{q_{2}\}$ and $\delta_{U}$ consists of the following transition rules: 
\begin{center}
  \begin{tabular}{ r  r  r }
     $\rightarrow^{n} q_{0}$ & $\rightarrow^{t} q_{1}$ &  $(q_{0},q_{0}) \rightarrow^{N} q_{0}$  \\
     $(q_{0},q_{0}) \rightarrow^{T} q_{1}$ & $(q_{0},q_{1}) \rightarrow^{N} q_{1}$ & 
     $(q_{1},q_{0}) \rightarrow^{N} q_{1}$\\
     $(q_{0},q_{1}) \rightarrow^{T} q_{2}$ & $(q_{1},q_{0}) \rightarrow^{T} q_{2}$ & 
     $(q_{1},q_{1}) \rightarrow^{T} q_{2}$\\
     $(q_{0},q_{2}) \rightarrow^{T} q_{2}$ & $(q_{2},q_{0}) \rightarrow^{T} q_{2}$ & 
     $(q_{1},q_{2}) \rightarrow^{T} q_{2}$\\
     $(q_{2},q_{1}) \rightarrow^{T} q_{2}$ & $(q_{2},q_{2}) \rightarrow^{T} q_{2}$ & 
     $(q_{1},q_{1}) \rightarrow^{N} q_{2}$\\
     $(q_{0},q_{2}) \rightarrow^{N} q_{2}$ & $(q_{2},q_{0}) \rightarrow^{N} q_{2}$ & 
     $(q_{1},q_{2}) \rightarrow^{N} q_{2}$\\
      $(q_{2},q_{1}) \rightarrow^{N} q_{2}$ & $(q_{2},q_{2}) \rightarrow^{N} q_{2}$ \\
    \end{tabular}
\end{center}

The set $\Phi$ of the following formulae presents a translation of the verification problem.  
We use the syntax of first-order logic used in   Mace4 finite model finder \cite{McCune}.

{\footnotesize

\begin{verbatim}



T(n,n,q0).
T(t,n,q1).
T(n,t,q3).
T(x,z,q0) & T(y,v,q0) -> T(fT(x,y),fN(z,v),q1).
T(x,z,q1) & T(y,v,q0) -> T(fN(x,y),fT(z,v),q2).
T(x,z,q0) & T(y,v,q1) -> T(fN(x,y),fT(z,v),q2).
T(x,z,q0) & T(y,v,q0) -> T(fN(x,y),fN(z,v),q0).
T(x,z,q0) & T(y,v,q2) -> T(fN(x,y),fN(z,v),q2).
T(x,z,q2) & T(y,v,q0) -> T(fN(x,y),fN(z,v),q2).
T(x,z,q3) & T(y,v,q0) -> T(fT(x,y),fN(z,v),q2). 
T(x,z,q0) & T(y,v,q3) -> T(fT(x,y),fN(z,v),q2). 
T(x,z,q0) & T(y,v,q0) -> T(fN(x,y),fT(z,v),q3). 

% Initial states automaton 

Init(n,q0).
Init(t,q1).
Init(x,q0) & Init(y,q0) -> Init(fT(x,y),q1). 
Init(x,q0) & Init(y,q1) -> Init(fN(x,y),q1).
Init(x,q0) & Init(y,q0) -> Init(fN(x,y),q0).
Init(x,q1) & Init(y,q0) -> Init(fN(x,y),q1).

% Bad states automaton 

Bad(n,q0).
Bad(t,q1).
Bad(x,q0) & Bad(y,q0) -> Bad(fN(x,y),q0).
Bad(x,q0) & Bad(y,q0) -> Bad(fT(x,y),q1).
Bad(x,q0) & Bad(y,q1) -> Bad(fN(x,y),q1).
Bad(x,q1) & Bad(y,q0) -> Bad(fN(x,y),q0).

Bad(x,q0) & Bad(y,q1) -> Bad(fT(x,y),q2).
Bad(x,q1) & Bad(y,q0) -> Bad(fT(x,y),q2).
Bad(x,q1) & Bad(y,q1) -> Bad(fN(x,y),q2).
Bad(x,q1) & Bad(y,q2) -> Bad(fT(x,y),q2).

Bad(x,q2) & Bad(y,q1) -> Bad(fT(x,y),q2).
Bad(x,q2) & Bad(y,q2) -> Bad(fT(x,y),q2).
Bad(x,q1) & Bad(y,q1) -> Bad(fN(x,y),q2).
Bad(x,q0) & Bad(y,q2) -> Bad(fN(x,y),q2).
Bad(x,q2) & Bad(y,q0) -> Bad(fN(x,y),q2).
Bad(x,q1) & Bad(y,q2) -> Bad(fN(x,y),q2).
Bad(x,q2) & Bad(y,q1) -> Bad(fN(x,y),q2).
Bad(x,q2) & Bad(y,q2) -> Bad(fN(x,y),q2).

T(x,y,q2) -> R(x,y).
R(x,y) & R(y,z) -> R(x,z).

Init(x,q1) -> Init1(x). 
Bad(x,q2) -> Bad1(x).


 
\end{verbatim}
}


According to Proposition~\ref{prop:encoding} and Corollary~\ref{cor:verification} to establish safety for Two-way Token protocol  it does suffice to show
$\Phi \not\vdash \exists x \exists y ((Init1(x) \land R(x,y)) \land Bad1(y))$. We delegate this task to Mace4 finite model finder and it finds a countermodel for $\Phi \rightarrow \exists x \exists y ((Init1(x) \land R(x,y)) \land  Bad1(y))$ in $0.03s$. The parameterized protocol is verified. Actual Mace4 input and output can be found in \cite{AL09}.

\section{Monotonic abstraction and symbolic reachability vc FCM}

Regular Tree Model Checking provides with a general method for the verification parameterized protocols for tree-shaped architectures. In \cite{PTS} a lightweight alternative to RTMC  was proposed. It utilizes a generic approach to safety verification using \emph{monotonic abstraction} and \emph{symbolic reachability} applied to tree rewriting systems.  This generic approach  has previously been successfully applied to the verification of parameterized linear system \cite{Mon} (as an alternative to standard Regular Model Checking).   In this section we demonstrate the flexibility of the FCM approach and show that one can translate safety verification problems for parameterized tree-shaped systems formulated using tree rewriting  into the problem of disproving a first-order formulae using the same basic principles (reachability as FO derivability). For defined translation we show the relative completeness of the FCM with respect to monotonic abstraction and symbolic reachability  and demonstrate its practical efficiency.   

\subsection{Parameterized Tree Systems}~\label{subsec:pts} 

The approach of \cite{PTS} to the verification of parameterized tree systems adopts the following viewpoint. A configuration of the system is represented by a tree over a finite alphabet, where elements of the alphabet represent the local states of the individual processes. The behaviors of the system is specified by a set of \emph{tree rewriting} rules, which describe  how the processes perform transitions. Transitions are enabled by the local states of the process together with the states of children and parent processes. 

\begin{definition}~\label{def:tree_over-states} 
A tree $T$ over a set of states  $Q$ is a pair $(S,\lambda)$, where 

\begin{itemize} 
\item $S$ is a tree structure (cf. Definition~\ref{def:tree}) 
\item $\lambda$ is a a mapping from $S$ to $Q$.  
\end{itemize} 
\end{definition}  

Notice that  trees over a set of states are similar to the trees over ranked alphabets (Definition~\ref{def:tree}) with the only difference is that the same state can label the vertices with  different number of children (e.g. leaves of the tree and internal vertices).  

In what follows to assume for simplicity of presentation (after \cite{PTS}) that all trees   are (no more than) binary, that is every node has either one or two  children (internal node) or no children (leaf).  It is straightforward to extend all constructions and results to the general case of not necessarily binary trees. Notice that configurations of the tree systems will be modeled by complete binary trees. Incomplete binary trees (which may contain nodes with one child) will  appear only in the rewrite rules.

\begin{definition}~\label{def:pts}
A parameterized tree system ${\cal P}$ is a tuple $(Q,R)$, where $Q$ is a finite set of states and $R \subseteq T(Q \times Q)$ is a finite set of rewrite rules.  
\end{definition}

For each rule $r = (S,\lambda) \in R$ we associate two trees, called \emph{left} and \emph{right} trees of $r$. We define $lhs(r) = (S,lhs(\lambda))$ and $rhs(r) = (S,rhs(\lambda))$, where $lhs(r)$ and $rhs(r)$ are left, respectively right projection of $\lambda$. 

We will denote (labeled) binary trees by bracket expressions in a standard way. 

\begin{example}
Let $Q = \{q_{0}, q_{1}, q_{2}\}$ then $r = \langle q_{0},q_{1}\rangle(\langle q_{1},q_{1} \rangle, \langle q_{2},q_{0}\rangle) \in T(Q \times Q)$ is a rewriting rule.  This rule has $\bullet (\bullet, \bullet)$ as it tree structure with one root and two leaves.  The pairs of states  $\langle q_{0},q_{1}\rangle$, $\langle q_{1},q_{1} \rangle$, $\langle q_{2},q_{0}\rangle$ label the root and two leaves respectively. We also have $lhs(r) = q_{0}(q_{1},q_{2})$ and $rhs(r) = q_{1}(q_{1},q_{0})$. 

\end{example}

\begin{example}
Let $Q$ be as above then $\langle q_{1},q_{2}\rangle (\langle q_{0},q_{1}\rangle)$ is a rewriting rule with the structure of incomplete binary tree $\bullet (\bullet)$ 
\end{example}

Given a parameterized tree system ${\cal P} = (Q,R)$ define one step transition relation $\Rightarrow_{\cal P} \subseteq 
T(Q) \times T(Q)$ as follows:   
%
$\tau_{1} \Rightarrow_{\cal P} \tau_{2}$ iff for some $r \in R$ $\tau_{1}$ contains $lhs(r)$ as a subtree and $\tau_{2}$ obtained from $\tau_{1}$ by replacing this subtree with $rhs(r)$. Since $lhs(r)$ and $rhs(r)$ have the same tree structure, the operation of replacement and one step transition relation are well-defined.   

\begin{example}
Let  ${\cal P} = (Q,R)$ with $Q = \{q_{0}, q_{1}, q_{2}\}$ and $R =\{\langle q_{0},q_{1}\rangle(\langle q_{1},q_{1} \rangle, \langle q_{2},q_{0}\rangle) \}$. Then we have (with the subtrees refered to in the definition of $\Rightarrow_{\cal P}$  inderlined):  
\begin{itemize}
\item $\underline{q_{0}(q_{1},q_{2})} \Rightarrow_{\cal P} \underline{q_{1}(q_{1},q_{0})}$;
\item $q_{2}(\underline{q_{0}(q_{1},q_{2})},q_{1}) \Rightarrow_{\cal P} q_{2}(\underline{q_{1}(q_{1},q_{0}}),q_{1})$;   
\item $\underline{q_{0}}(\underline{q_{1}}(q_{1},q_{0}),\underline{q_{2}}(q_{0},q_{2})) \Rightarrow_{\cal P} 
\underline{q_{1}}(\underline{q_{1}}(q_{1},q_{0}),\underline{q_{0}}(q_{0},q_{2})) $; 
\end{itemize}

\end{example}

\noindent 
Denote transitive and reflexive closure of $\Rightarrow_{\cal P}$ by $\Rightarrow_{\cal P}^{\ast}$. 

\begin{definition}~\label{def:embedding}(embedding)
For $\tau_{1} = (S_{1},\lambda_{1})$ and $\tau_{2} = (S_{2},\lambda_{2})$ an injective  function $f: S_{1} \rightarrow S_{2}$ is called embedding iff 
\begin{itemize}
\item $s\cdot b \in S$ implies $f(s) \cdot b  \le f(s \cdot b)$ for any $s \in S$ 
\item $\lambda_{1}(s) = \lambda_{2}(f(s))$
\end{itemize}
\end{definition}

We use $\tau_{1} \preceq_{f} \tau_{2}$ to denote that $f$ is embedding of $\tau_{1}$ into $\tau_{2}$ and write $\tau_{1} \preceq \tau_{2}$ iff there exists $f$ such that $\tau_{1} \preceq_{f} \tau_{2}$. 

Using embeddability relation $\prec$ allows to describe infinite families of trees by \emph{finitary} means.

We call a set of trees $T \subseteq T(Q)$ \emph{finitely based} iff there is a finite set $B \subseteq T(Q)$ such that   
$T = \{ \tau \mid \exists \tau' \in B  \tau' \preceq \tau \}$. Notice that finitely based set of trees are upwards closed with respect to $\preceq$, that is $\tau \in T$ and $\tau \preceq \tau'$ implies  $\tau' \in T$.

Many safety verification problems for  parameterized tree system can be reduced to the following  coverability problem.

\begin{problem}~\label{prob:mono}
Given a parameterized tree system ${\cal P} = (Q,R)$, a regulat tree language $Init  \subseteq T(Q)$ of initial configurations   and finitely based set of unsafe configurations 
$Unsafe \subseteq T(Q)$.  Does $\tau \not\Rightarrow_{\cal P}^{\ast} \tau'$ hold for all $\tau \in Init$ and all 
$\tau' \in Unsafe$? 
\end{problem} 

\begin{note} We formally defined regular tree languages over ranked alphabets. Regular tree languages over (unranked) states can be defined in a various ways. We will fix a particular convention in Assumption 1 below.  

\end{note}

Now we briefly outline the monotonic abstraction approach \cite{PTS} to verification. Given the coverability problem above  \cite{PTS} defines the monotonic abstraction $\Rightarrow_{\cal P}^{\cal A}$ of the transition relation $\Rightarrow_{\cal P}$ as follows. We have $\tau_{1} \Rightarrow_{\cal P}^{\cal A} \tau_{2}$ iff there exists a tree $\tau'$ such that $\tau' \preceq \tau_{1}$ and $\tau' \Rightarrow_{\cal P} \tau_{2}$. It is clear that such defined $\Rightarrow_{\cal P}^{\cal A}$ is an over-approximation of  $\Rightarrow_{\cal P}$. To establish the safety property, i.e. to get a positive answer to the  question of Problem~\ref{prob:mono}, \cite{PTS} proposes using a symbolic backward reachability algorithm for monotonic abstraction. Starting with an \emph{upwards} closed (wrt to $\prec$) set of unsafe configuration $Unsafe$ the algorithm proceeds iteratively with the computation of the set of configurations backwards reachable along  $\Rightarrow_{\cal P}^{\cal A}$ from $Unsafe$: 

\begin{itemize}
\item $U_{0} = Unsafe$
 \item $U_{i+1} = U_{i} \cup Pre(U_{i})$ 
\end{itemize}   

where $Pre(U) = \{\tau \mid \exists \tau' \in U \land \tau \Rightarrow_{\cal P}^{\cal A} \tau' \}$.  
Since the relation $\preceq$ is a \emph{well quasi-ordering} \cite{Kruskal:tree} this iterative process is guaranteed to stabilize, i.e. $U_{n+1} = U_{n} = U$ for some finite $n$. During the computation each $U_{i}$ is represented symbolically by a finite set of generators. Once the process stabilized on some $U$ the check is preformed on whether 
$Init \cap U = \emptyset$. If this condition is satisfied then the safety is established, for no bad configuration can be reached from initial configurations via $\Rightarrow-{\cal P}^{\cal A}$ and, a fortiori, via $\Rightarrow_{\cal P}$.

\subsection{Parameterized Tree systems to FCM}

Here we show how to translate the coverability problem (Problem~\ref{prob:mono}) into the task of disproving a first-order formula and demonstrate the \emph{relative  completeness} of the FCM method with respect to monotonic abstraction approach.  

Assume we are given an instance of the coverability problem, that is  
\begin{itemize}
\item a parameterized tree system ${\cal P} = (Q,R)$, 
\item a regular tree language $Init$ of initial configurations, given by a tree automaton $A_{I} = (Q_{I},F,\delta)$, and     
\item finitely based set of unsafe configurations $Unsafe$ given by a finite set of generators $Un  \subseteq T(Q)$.  

\end{itemize}

\noindent
For a set of states $Q$  let ${\cal F}_{Q} = \{ f_{q}^{(2)} \mid q \in Q\} \cup \{e\}$ be the set of corresponding binary functional symbols extended with a distinct functional symbol $e$  of arity $0$ (constant).

For any \emph{complete} binary tree $\tau \in T(Q)$ define its term translation $t_{\tau}$ in vocabulary ${\cal F}_{Q}$ inductively: 

\begin{itemize} 
\item  $t_{\tau} = f_{q}(e,e)$ if $\tau$ is a tree with one node labeled by a state $q$; 
\item $t_{\tau} = f_{q}(t_{\tau_{1}},t_{\tau_{2}})$ if the root of $\tau$ has two children and  $\tau = q(\tau_{1},\tau_{2})$; 
\end{itemize}

\noindent
For any not necessarily complete binary tree $\tau \in T(Q \times Q)$ define inductively  its translation $s_{\tau}$  as a set of pairs of terms in vocabulary  ${\cal F}_{Q}$: 

\begin{itemize}
\item $s_{\tau} = \{\langle f_{q_{1}}(e,e), f_{q_{2}}(e,e) \rangle \}$ if $\tau$ is a tree with one node labeled with states $(q_{1},q_{2})$; 
\item $s_{\tau} =\{\langle f_{q_1}(\rho_{1}, \rho_{2}),f_{q_2}(\rho_{3},\rho_{4}) \rangle \mid \langle \rho_{1},\rho_{3} \rangle \in s_{\tau_{1}}, \langle \rho_{2},\rho_{4} \rangle \in s_{\tau_{2}}\}$ if the root of $\tau$ is labeled by $(q_{1},q_{2})$ and it has two children $\tau_{1}$ and $\tau_{2}$, i.e. if  $\tau = (q_1,q_2)(\tau_{1}, \tau_{2})$;
\item $s_{\tau} = \{\langle f_{q_1}(\rho_{1},e),f_{q_2}(\rho_{2},e)\rangle \mid \langle \rho_{1},\rho_{2} \rangle \in s_{\tau_{1}}\} \cup \{\langle f_{q_1}(e,\rho_{1}),f_{q_2}(e,\rho_{2})\rangle \mid \langle \rho_{1},\rho_{2} \rangle \in s_{\tau_{1}} \}$ if the root of $\tau$ is labeled by $(q_{1},q_{2})$ and it has one child $\tau_{1}$, i.e. if $\tau = (q_{1},q_{2})(\tau_1)$. 
\end{itemize}

\noindent
For $\langle \rho_{1},\rho_{2} \rangle \in s_{\tau}$ we denote by $\rho_{1}^{gen}$ (by  
$\rho_{2}^{gen}$) a generalized term obtained by replacement of all occurences of constant $e$ in $\rho_{1}$ (in $\rho_{2}$, respectively,)   with distinct variables.

Now we define first-order translation of the set of rules $R$ as the following set $\Phi_{R}$ of first-order  formulae, which are all assumed to be universally closed: 

\begin{enumerate}
\item $R(\rho_{1}^{gen},\rho_{2}^{gen})$ for all $r \in R$ and $\langle \rho_{1},\rho_{2} \rangle \in s_{r}$ \hspace*{13mm}   {\bf rewriting axioms}  
\item $R(x,x)$       \hspace*{69mm}  {\bf reflexivity axiom} 
\item $R(x,y) \land R(y,z) \rightarrow R(x,z)$  \hspace*{38mm}  {\bf transitivity axiom}
\item $R(x,y) \land R(z,v) \rightarrow R(f_{q}(x,z),f_{q}(y,v))$ for all $q \in Q$ \\ \hspace*{80mm} {\bf congruence axioms}                        
\end{enumerate}

In $1)$ we additionally require that generalizations $\rho_{1}^{gen}$ and $\rho_{2}^{gen}$ should be consistent, 
that means the variables used in the generalizations are the same in the same positions. 

Now for simplicity  we make the following  
\begin{assumption}
An automaton $A_{I} = (Q_{I},F_{I},\delta_{I})$ is given over ranked alphabet ${\cal F}_{Q}$. 
\end{assumption} 

We define the  translation of $A_{I}$ as the set $\Phi_{I}$ of first-order formulae

\begin{enumerate}
\setcounter{enumi}{4}
\item $I_{\theta}(f_{q}(e,e))$ for all $\rightarrow^{e} \theta'$ and $(\theta',\theta') \rightarrow^{f_{q}} \theta$ in $\delta_{I}$;  
\item $I_{\theta_{1}}(x) \land I_{\theta_{2}}(y) \rightarrow I_{\theta_{3}}(f_{q}(x,y))$ for all $(\theta_{1},\theta_{2}) \rightarrow^{f_{q}} \theta_{3})$ in $\delta_{I}$. 
\item $\vee_{\theta \in F_{I}} I_{\theta}(x)  \rightarrow Init(x)$ 
\end{enumerate}

Let $A_{U} = (Q_{U},F_{U}, \delta_{U})$ is a tree automaton recognizing finitely based set $Unsafe$. Then its translation $\Phi_{U}$ defined analogously to the translation of $A_{I}$: 

\begin{enumerate}
\setcounter{enumi}{7}
\item $U_{\theta}(f_{q}(e,e))$ for all $\rightarrow^{e} \theta'$ and $(\theta',\theta') \rightarrow^{f_{q}} \theta$ in $\delta_{U}$;  
\item $U_{\theta_{1}}(x) \land U_{\theta_{2}}(y) \rightarrow U_{\theta_{3}}(f_{q}(x,y))$ for all $(\theta_{1},\theta_{2}) \rightarrow^{f_{q}} \theta_{3})$ in $\delta_{U}$. 
\item $\vee_{\theta \in F_{U}} U_{\theta}(x)  \rightarrow Unsafe(x)$ 
\end{enumerate}

\begin{proposition}(Adequacy of encoding)
For an  instance of the coverability problem and the translation defined above the following holds true: 
\begin{enumerate}
\item For any $\tau_{1}, \tau_{2} \in T(Q)$ if $\tau_{1} \Rightarrow_{\cal P}^{\ast} \tau_{2}$ then $\Phi_{R} \vdash R(t_{\tau_{1}},t_{\tau_{2}})$
\item For any $\tau \in Init$ $\Phi_{I} \vdash Init(t_{\tau})$;  
\item For any $\tau \in Unsafe$ $\Phi_{U} \vdash Unsafe(t_{\tau})$
\end{enumerate}
\end{proposition}

\begin{proof} proceeds by straightforward inspection of definitions. 
\end{proof}
\begin{corollary}(safety verification)
If $\Phi_{R} \cup \Phi_{I} \cup \Phi_{U} \not\vdash \exists  x \exists y  Init(x) \land Unsafe(y) \land R(x,y)$ then 
the coverability problem has a positive answer, that is $\tau \not\Rightarrow_{\cal P}^{\ast} \tau'$ holds for all $\tau \in Init$ and all 
$\tau' \in Unsafe$.   
\end{corollary}

\begin{theorem}(relative completeness)
Given a parameterized tree system ${\cal P} =   (Q,R)$, the tree regular language of initial configurations $Init$, 
finitely based set of unsafe configurations $Unsafe$. Assume the backward symbolic reachability algorithm for monotonic abstraction described above 
terminates with the fixed-point $U = U_{n+1} = U_{n}$ for some $n$ and  $Init \cap U = \emptyset$. Then there exists a finite model for $\Phi_{R} \land  \Phi_{I} \land  \Phi_{U} 
\land \neg (\exists  x \exists y  Init(x) \land Unsafe(y) \land R(x,y))$.   
\end{theorem}
\begin{proof} First we observe that since the fixed-point $U$ has a finite set of generators it is a regular tree 
language. Let $A_{U^{\ast}} = (Q_{U^{\ast}},F_{U^{\ast}},\delta_{U^{\ast}})$ be a deterministic tree automaton 
recognizing $U$. 
We take $Q_{U^{\ast}}$ as a domain of the required model. 
Interpretations of all functional symbols from ${\cal F}_{Q}$ are given by $\delta_{U^{\ast}}$: 
\begin{itemize}
\item $[f_{q}](\theta_{1},\theta_{2}) = \theta_{3}$ iff $(\theta_{1},\theta_{2}) \rightarrow^{f_{q}} \theta_{3}$ is in $\delta_{U^{\ast}}$ 
\item $[e] = \theta$, where $\rightarrow^{e} \theta$ is in $\delta_{U^{\ast}}$. 
\end{itemize} 
Interpretations of predicates $R, I_{\theta}, Init, U_{\theta}, Unsafe$ are defined inductively as the least sets of 
tuples, or elements of the domains satisfying the axioms 1-4, 5-7, 8-10, respectively. 
That concludes the definition of the model which we denote by ${\cal M}$. 
We have ${\cal M} \models  \Phi_{R} \land \Phi_{I} \land \Phi_{U}$ by construction. 
Now we check that ${\cal M} \models \neg (\exists  x \exists y  Init(x) \land Unsafe(y) \land R(x,y))$ 
is satisfied in the model. We have 
\begin{enumerate}
\item $[Init]\star[R] \subseteq \{[\tau] \mid \exists \tau' \in Init \; \tau' \Rightarrow^{\ast}_{\cal P} \tau \}$ (by the minimality conditions on interpretations of $Init$ 
and $R$)
\item $\{[\tau] \mid \exists \tau' \in Init \; \tau' 
\Rightarrow^{\ast}_{\cal P} \tau \} \subseteq \bar{F}_{U^{\ast}} = Q - F_{U^{\ast}}$ (by assumption $U \cap Init = \emptyset$)
\item $[Unsafe] \subseteq F_{U^{\ast}}$ (by $Unsafe \subseteq U$); 
\item $([Init]\star[R]) \cap [Unsafe] = \emptyset$ (by 1-3). 
\end{enumerate}
\end{proof}

\subsection{The case study, II}

In this section we illustrate the discussed variation of the FCM method by applying it again  to the verification of Two-way Token Protocol, but specified differently.  The specification of this protocol using trees over states and tree rewriting is taken from \cite{PTS}. The set of states $Q = \{n,t\}$, where $n$ and $t$  denote local states `no token' and `token', respectively. The set of $R$ of rewriting rules consists of the following rules: 
\begin{itemize}
\item $\langle t,n \rangle(\langle n,t \rangle)$; 
\item $\langle n,t \rangle(\langle t,n \rangle)$; 
\end{itemize}

The set $Init$ of initial configurations consists all complete binary trees over $Q$ with exactly one token. 
The set $Unsafe$ of unsafe configuration consists of all complete binary trees over $Q$ with 
at least two tokens. The set of the formulae $\Phi$ below is a first-order translation (in Mace4 syntax) of the verification task.


{\footnotesize
\begin{verbatim}

% rewriting rules

R(ft(fn(y,z),x),fn(ft(y,z),x)). 
R(ft(x,fn(y,z)),fn(x,ft(y,z))). 
R(fn(ft(y,z),x),ft(fn(y,z),x)). 
R(fn(x,ft(y,z)),ft(x,fn(y,z))). 

% reflexivity

R(x,x).

%congruence
(R(x,y) & R(z,v)) -> R(fn(x,z),fn(y,v)). 
(R(x,y) & R(z,v)) -> R(ft(x,z),ft(y,v)).

% transitivity

(R(x,y) & R(y,z)) -> R(x,z). 

% Initial states automaton 
I1(fn(e,e)). 
(I1(x) & I1(y)) -> I1(fn(x,y)). 
(I1(x) & I1(y)) -> Init(ft(x,y)). 
(Init(x) & I1(y)) -> Init(fn(x,y)).
(I1(x) & Init(y)) -> Init(fn(x,y)). 

% Unsafe states automaton
B1(ft(x,y)).
B1(x) -> B1(fn(x,y)). 
B1(y) -> B1(fn(x,y)).  
B1(x) -> Unsafe(ft(x,y)). 
B1(x) -> Unsafe(ft(y,x)).
B1(x) & B1(y) -> Unsafe(fn(x,y)).
B1(x) & B1(y) -> Unsafe(ft(x,y)).  
Unsafe(x) -> Unsafe(fn(x,y)). 
Unsafe(x) -> Unsafe(fn(y,x)). 
Unsafe(x) -> Unsafe(ft(x,y)). 
Unsafe(x) -> Unsafe(ft(y,x)).



\end{verbatim}
}

Now, in order to establish safety, it is sufficient to show that $\Phi \not\vdash \exists x \exists y Init(x) \land R(x,y) \land Unsafe(y)$. Finite model finder Mace4 finds a model for $\Phi \land \neg(\exists x \exists y Init(x) \land R(x,y) \land Unsafe(y))$ in $0.04s$.

\section{Experimental results}

We have applied both presented versions of FMC method to the verification  of several parameterized tree-shaped systems. 
The tasks specified in RTMC tradition were taken from \cite{RTMC} and the first translation was used. To compare with monotonic abstraction based methods we used the second translation for the tasks from \cite{PTS}.   

\noindent 
In the experiments we used the finite model finder Mace4\cite{McCune} within the package 
Prover9-Mace4, Version 0.5, December 2007. 
The system configuration used in the experiments:   Microsoft Windows XP Professional, Version 2002, Intel(R) Core(TM)2 Duo CPU, T7100 @ 1.8Ghz  1.79Ghz,  1.00 GB of  RAM.  The time measurements are done by Mace4 itself, upon completion 
of the model search it communicates the CPU time used. The table below lists the parameterized tree  
protocols 
and shows the time it took Mace4 to 
find a countermodel and verify a safety property.  
The time shown is an average of 10 attempts. We also show the time reported on the verification of the same protocols by alternative methods. 

\subsection{FCM vs RTMC}

\begin{center}
  \begin{tabular}{| l | r | r | }
    \hline
    Protocol & Time & Time reported in \cite{ARTMC}$^{\ast}$\\ \hline
    Token & 0.02s & 0.06s \\ \hline
    Two-way Token  & 0.03s & 0.09s \\ \hline
  \end{tabular}
\end{center}

$^{\ast}$ the system configuration  used in \cite{ARTMC} was \emph{Intel Centrino 1.6GHZ with 768MB of RAM}

Notice that \cite{ARTMC} discusses different methods for enhancement of RTMC within the abstract-check-refine paradigm and we included in the table the best times reported in \cite{ARTMC} for each verification problem.

\subsection{FCM vs monotonic Abstraction} 

\begin{center}
  \begin{tabular}{| l | r | r | }
    \hline
    Protocol & Time & Time reported in \cite{PTS}$^{\ast}$\\ \hline
    Token & 0.02s & 1s \\ \hline
    Two-way Token  & 0.03s & 1s \\
    \hline
    Percolate & 0.02s & 1s \\
\hline
  Leader Election & 0.03s & 1s\\
\hline
    Tree-arbiter & 0.02s & 37s \\
\hline 
IEEE 1394 & 0.04s & 1h15m25s\\
\hline
  \end{tabular}
\end{center}

$^{\ast}$ the system configuration  used in \cite{PTS} was \emph{dual Opteron 2.8 GHZ with 8 GB of RAM}

All specifications used in the experiments and Mace4 output can be found in \cite{AL09}. 

\section{Related work}~\label{sec:rel}

As mentioned Section 1 the approach to verification using the modeling of protocol executions by 
first-order derivations and together with  countermodel finding for disproving was introduced within the research on the 
formal analysis of cryptographic protocols (\cite{Weid99},\cite{S01},\cite{GL08}, \cite{JW09}, \cite{Gut}).

This work continues the exploration of the FCM approach presented in \cite{AL09,Avocs09,ALWIng10,AL10,AL10arxiv,AL11}. In \cite{AL10}(which is an extended version of \cite{Avocs09}) it was shown that FCM provides a decision procedure for safety verification for lossy channel systems, and that FCM can be used for efficient verification of parameterised cache coherence protocols.  The relative completeness of the FCM with respect to  regular model checking and methods based on monotonic abstraction for linear parameterized systems was established in \cite{AL10arxiv}(which is an extended version of the abstract \cite{ALWIng10}). The relative completeness of the FCM with respect to tree completion techniques for general term rewriting systems is shown in \cite{AL11}. Our treatment of tree rewriting in \ref{subsec:pts}  can be seen as a particular case of term rewriting considered in \cite{AL11} with slightly different  translation of tree automata. 
Detailed comparison and/or  unified treatment of FCM vs Tree Completion vs RTMC vs Monotonic Abstraction to be given elsewhere.  Here we notice only that the reason for FCM to succeed in verification of safety of various classess of infinite-state and parameterized systems is the presence of \emph{regular} sets of configurations (invariants) covering all reachable configurations and disjoint with the sets of unsafe configurations.

In a more general context, the work  we present in this paper is related to the concepts of \emph{proof by consistency} \cite{pbc}, 
and  \emph{inductionless induction} \cite{ii} and can be seen as an investigation into the power of these concepts in 
the particular setting of the verification of parameterized tree systems via finite countermodel finding.
          
\section{Conclusion} 

We have shown how to apply generic finite model finders in the parameterized verification of tree-shaped systems, have demonstrated the relative completeness of the method with respect to regular tree model checking and to the methods based on monotonic abstraction and have illustrated its practical efficiency. Future work includes the investigation of scalability of FCM,  and its applications to software verification.

\section*{Acknowledgments} 

The author is grateful to anonymous referees of FMCAD 2011 conference who provided with many helpful comments on the previous version of this paper.


\begin{thebibliography}{}

 

\bibitem{Mon} Abdulla, P.A., Delzanno G., Henda, N.B., \& Rezine, A., (2009a),
\newblock  Monotonic Abstraction: on Efficient Verification of Parameterized Systems. 
\newblock {\em Int. J. Found. Comput. Sci.} 20(5): 779-801 


\bibitem{Ab} Abdulla, P.A., Jonsson,  B., (1996)
\newblock Verifying programs with unreliable channels. 
\newblock {\em Information and Computation}, 127(2):91-101, 1996. 

\bibitem{RTMC} Abdulla, Parosh Aziz and Jonsson, Bengt and Mahata, Pritha and d'Orso, Julien, Regular Tree Model Checking,
in Proceedings of the 14th International Conference on Computer Aided Verification, CAV '02, 2002, 555--568
 

\bibitem{RMCs}
Abdulla, P.A., Jonsson,B., Nilsson, M., \& Saksena, M., (2004) A Survey of Regular Model Checking, In Proc. of CONCUR'04, volume 3170 of LNCS, pp 35--58, 
2004. 

\bibitem{PTS}
Parosh Aziz Abdulla, Noomene Ben Henda, Giorgio Delzanno, Frédéric Haziza and Ahmed Rezine, 
Parameterized Tree Systems, in FORMAL TECHNIQUES FOR NETWORKED AND DISTRIBUTED SYSTEMS – FORTE 2008, 
Lecture Notes in Computer Science, 2008, Volume 5048/2008, 69-83. 


\bibitem{ARTMC}
A.~Bouajjani, P.~Habermehl,A.~Rogalewicz, T.~Vojnar, 
Abstract Regular Tree Model Checking, Electronic Notes in Theoretical Computer Science, 149, (2006), 37--48.  


\bibitem{Model} 
Caferra, R.,  Leitsch, A., \& Peltier, M., (2004) {\em Automated Model Building}, Applied Logic Series, 31,  Kluwer, 2004. 

\bibitem{ii}
Comon, H., (1994), Inductionless induction. 
In R.~David, ed. \emph{2nd Int. Conf. in Logic for Computer Science: Automated Deduction. Lecture Notes}, 
Chambery, Uni de Savoie, 1994.  








\bibitem{GL08} Goubault-Larrecq, J., (2010), Finite Models for Formal Security Proofs, \emph{Journal of Computer Security}, 6: 1247--1299, 2010. 




\bibitem{Gut} Guttman, J., (2009)  Security Theorems via Model Theory, Proceedings 16th International Workshop on Expressiveness in Concurrency, EXPRESS, EPTCS, vol. 8 (2009)




\bibitem{JW09}
Jurjens, J.,  \& Weber, T., (2009),  Finite Models in FOL-Based
Crypto-Protocol Verification,  P. Degano and L. Vigan`o (Eds.): ARSPA-WITS 2009, LNCS 5511, pp. 155-172, 2009.


\bibitem{pbc} Kapur, D.,  \& Musser, D.R., (1987),  Proof by consistency. \emph{Artificial Intelligence}, 31:125--157, 1987. 

\bibitem{Kruskal:tree} Kruskal, J.B., (1960), Well-Quasi-Ordering, the Tree Theorem, and {V}azsonyi's Conjecture, \emph{Transactions of the American Mathematical Society}, 95,2:210--225, 1960. 

\bibitem{AL09}
Lisitsa, A., (2009a), 
\newblock Verfication via countermodel finding\\
\verb"http://www.csc.liv.ac.uk/~alexei/countermodel/"

\bibitem{Avocs09} Lisitsa, A., (2009b), 
\newblock Reachability as deducibility, finite countermodels and verification. 
\newblock In preProceedings of AVOCS 2009, 
Technical Report of Computer Science, Swansea University, CSR-2-2009, pp 241-243. 

\bibitem{ALWIng10} 
Lisitsa, A., (2010a), 
\newblock Finite countermodels as invariants. A case study in verification of parameterized mutual exclusion protocol. 
In Proceedings of WING 2010, 1pp

\bibitem{AL10} Lisitsa, A., (2010b), 
\newblock Reachability as deducibility, finite countermodels and verification. 
\newblock In Proceedings of ATVA 2010, LNCS 6252, 233--244



\bibitem{AL10arxiv} Lisitsa, A., (2010c), Finite model finding  for parameterized verification,  
CoRR abs/1011.0447: (2010)

\bibitem{AL11} Lisitsa, A., (2011), 
\newblock Finite models vs tree automata in safety verification, CoRR abs/1107.0349: (2011) .  



\bibitem{McCune}
McCune, W.,
\newblock Prover9 and Mace4
\verb"http://www.cs.unm.edu/~mccune/mace4/"


\bibitem{RMC} Nilsson, M., (2005)  Regular Model Checking. Acta Universitatis Upsaliensis. Uppsala
Dissertations from the Faculty of Science and Technology 60. 149 pp. Uppsala. ISBN
91-554-6137-9, 2005. 

\bibitem{S01} Selinger, P., (2001),   Models for an adversary-centric protocol logic. Electr. Notes Theor.
Comput. Sci. 55(1) (2001)


\bibitem{Weid99} Weidenbach, C., (1999),  Towards an Automatic Analysis of Security
Protocols in First-Order Logic, in H. Ganzinger (Ed.): CADE-16, LNAI 1632, pp. 314--328, 1999.





\end{thebibliography}
\end{document}